\documentclass[titlepage,12pt]{article}\usepackage[]{graphicx}\usepackage[]{color}
\makeatletter
\def\maxwidth{ %
  \ifdim\Gin@nat@width>\linewidth
    \linewidth
  \else
    \Gin@nat@width
  \fi
}
\makeatother

\definecolor{fgcolor}{rgb}{0.345, 0.345, 0.345}

\usepackage{framed}
\makeatletter
 {\par\unskip\endMakeFramed%
 \at@end@of@kframe}
\makeatother

\definecolor{shadecolor}{rgb}{.97, .97, .97}
\definecolor{messagecolor}{rgb}{0, 0, 0}
\definecolor{warningcolor}{rgb}{1, 0, 1}
\definecolor{errorcolor}{rgb}{1, 0, 0}
\newenvironment{knitrout}{}{} % an empty environment to be redefined in TeX

\usepackage{alltt}

\usepackage[margin=1in]{geometry}
\usepackage{amsmath,amsfonts,amsthm}
\usepackage{setspace}
\usepackage{multirow}
\usepackage{natbib}
\usepackage{bm}
\usepackage{booktabs}

\usepackage{array}
\usepackage{wrapfig}
\usepackage{float}
\usepackage{colortbl}
\usepackage{pdflscape}
\usepackage{tabu}
\usepackage{threeparttable}
\usepackage{threeparttablex}
\usepackage{graphicx}
\usepackage{url}

\theoremstyle{plain}% Theorem-like structures provided by amsthm.sty

\theoremstyle{definition}

\theoremstyle{remark}

\newcommand{\dalphaU}{d^{max}_\alpha}
\newcommand{\dalphaB}{d^{min}_\alpha}
\newcommand{\dstar}{d^*}

\newcommand{\dhatU}{d^{max}}
\newcommand{\dhatB}{d^{min}}
\newcommand{\dhatm}{\hat{d}_M}
\newcommand{\dhatmab}{\hat{d}^{ab}_M}
\newcommand{\EE}{E}

\newtheorem{prop}{Proposition}

\newtheorem{assumption}{Assumption}

\newtheorem{assumptionalt}{Assumption}[assumption]

\newenvironment{assumptionp}[1]{
  \renewcommand\theassumptionalt{#1}
  \assumptionalt
}{\endassumptionalt}

\DeclareMathOperator*{\argmax}{arg\,max}
\DeclareMathOperator*{\argmin}{arg\,min}

\providecommand{\keywords}[1]
{
  \small
  \textbf{\textit{Keywords---}} #1
}

\title{Sequential Specification Tests to Choose a Model: A Change-Point Approach}
\author{Adam C Sales\thanks{The author received valuable input from Atul Mallik, George Michailidis, Ben B. Hansen, James E. Pustejovsky, and an anonymous reviewer}}
\date{%
Worcester Polytechnic Institute\\
Department of Mathematical Sciences\\
Worcester, Massachusetts, USA\\
asales@wpi.edu}
\IfFileExists{upquote.sty}{\usepackage{upquote}}{}
\begin{document}
\maketitle

\begin{abstract}
  A researcher choosing between models ordered by some criterion may seek the best specification that satisfies a testable assumption. In this scenario, sequential specification tests (SSTs) are hypothesis tests of that assumption for each model in the sequence. We introduce a method using the p-values from SSTs to estimate the point in the sequence where the assumption ceases to hold. Unlike alternative approaches, this method is robust to individual errant p-values and does not require choosing a test level or tuning parameter. We demonstrate the method's properties with a simulation study, and illustrate it by choosing a bandwidth in a regression discontinuity design and a lag order for a time series model.\\

\end{abstract}

\section{Introduction}

Null hypothesis tests and p-values play a central role in model checking.
In this context, the null hypothesis may be that that the data are
drawn from a distribution contained in the the model under study, or
it may be derived from an underlying assumption.
Typically, researchers use these specification tests to check the fit
of a model chosen by other means, but in some cases hypothesis tests
form the basis of a model selection procedure.
In these cases, researchers construct a sequence of model
specifications, ordered by preferability, and test each one.
The best model whose assumptions ``pass'' the hypothesis test is chosen.

For example, take the datasets displayed in Figure \ref{fig:example},
which will be discussed in more detail in Section \ref{sec:examples}.
Figure \ref{fig:example}A shows the annual total unemployment rate in the
United States from 1890 to 2015.
One of the simpler models for time series such as these is an order
$p$ autoregression, or $AR(p)$, under which the value of the time
series at point $t$ may depend on its historical values at
$t-1,...,t-p$ but, conditional on those, is independent of values at
points before $t-p$.
To choose the order $p$, researchers may test model fit for a
sequence of lag orders $p$, and choose the smallest $p$ that the tests
fail to reject.
Here a smaller lag orders $p$ are preferable because they lead to more
parsimonious models and more precise estimates.

Figure \ref{fig:example}B plots data that \citet{lso} used to estimate the effect of academic
probation on college students' subsequent grade point averages.
University students were put on academic
probation if their first-year cumulative grade point averages
fell below a cutoff.
This is an example of a regression discontinuity design \citep[RDD;][]{thistlewhiteCampbell}, in
which treatment is assigned if a
numeric ``running variable'' $R$ falls below (or above) a
pre-specified cutoff $c$.
Since treatment assignment is entirely a function of $R$ and $c$,
researchers can model the relationship between $R$ and an outcome
variable $Y$ in
order to estimate the effect of the treatment without confounding.
A common tool for ensuring that RDD models are well specified is to
limit the data analysis sample to subjects with $R\in\{c-\omega,c+\omega\}$,
where $\omega>0$ is a bandwidth selected by the data analyst.
One method for choosing $\omega$ relies on subjects' baseline
covariates: researchers will estimate ``effects'' of the treatment on
baseline covariates using data from subjects with $R$ within $\omega$ of $c$.
Since the treatment cannot possibly have an effect on baseline
covariates, any estimated effects are due to model misspecification or
an overly-large choice of $\omega$.
Following this reasoning, some methodologists recommend testing for
effects on covariates using an array of candidate bandwidths, and
choosing the largest bandwidth within which the null hypothesis of no
effect cannot be rejected.
The bandwidth tradeoff is similar to the $AR(p)$ case: if $\omega$ is too
large, the causal model might be misspecified and the effect estimate
will be biased. If $\omega$ is too small, there will not be enough data to
precisely estimate the effect of interest.

These are both examples of the use of sequential specification tests
(SSTs) to choose a model.
SSTs are also used in covariate selection for regression models
\citep{greene2003econometric}, selecting the number of components in mixture models,
latent class analysis, and factor analysis \citep{nylund2007deciding} and in
propensity-score matching \citep{hansen2015}.

Do hypothesis tests make any sense in model selection?
The results of a null hypothesis test, of course, are never evidence
in favor of a null hypothesis; null hypotheses can only be rejected,
not accepted.
Along similar lines, the logic of controlling type-I error rates seems backwards when it comes to
model selection, in which accepting a
problematic specification---a type II error---is the major concern.
These issues have prompted some methodologists \citep[e.g.][]{cft} to
propose adjusting the size of specification tests to a value higher
than the conventional $\alpha=0.05$.
However, the appropriate value for $\alpha$, and the
criteria for selecting $\alpha$, remain unclear.

On the other hand, a conceptually-sound model-selection method based
on SSTs would be particularly useful;
specification tests already exist for most common models, and
they are regularly taught in introductory quantitative methods
classes.

\begin{figure}

\includegraphics[width=\maxwidth]{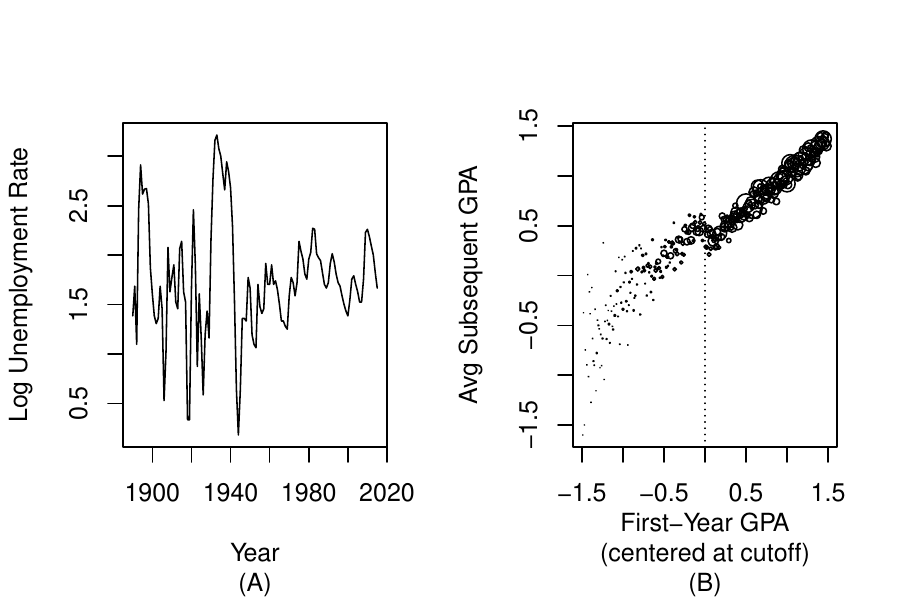}

\caption{Plot (A) shows a
  time-series of log annual United States total unemployment from 1890 to
  2015. Data were combined from \citet{urca} and \citet{cps}. Plot (B) shows data from
  \citet{lso}: average subsequent grade point averages (GPAs), as a function of first-year
  GPAs,
  centered at the academic probation cutoff (dotted line).%, with
  %least-squares regression fits.
  The points are sized proportionally to the number of students
  with each first-year GPA.}
\label{fig:example}
\end{figure}
%%% NOTE: schotman & van dijk used AR(4) for unemployment

This paper develops such a method, based on a clever idea in change-point or threshold
estimation.
\citet{mallik} points out that in a process with a change point,
the p-values from a sequence of tests of a null regression function
are uniformly-distributed as long as the regression function is
correct, but asymptotically zero when the function is not correct.
They use this dichotomous behavior to construct a simple, consistent
estimator of the change-point, that is, the point at which the null model
stops being correct.

We adapt that idea to the case of SSTs, choosing the change-point in
a sequence of models, i.e. the point when models stop being correct.
Our change-point estimator is based on the entire sequence of p-values, so that (unlike under current approaches)
an individual outlier p-value will not drive its conclusions.
%% Thus, the change-point view of model selection is arguably conceptually
%% more satisfying and practically more reliable than the conventional
%% test-based approach.
What's more, unlike other SST model selectors, the change-point
approach does not require the researcher to specify a level $\alpha$
or any other tuning parameter.
This approach shifts the model selection rationale away
from the logic of hypothesis testing, based on type-I and type-II
error rates, and towards the logic of estimation.

Model selection is a very broad field in statistics with a rich history and---since an appropriate model must be chosen before any data analysis can proceed---it is centrally important.
A nice overview can be found in \citet{rao2001model}\footnote{This very helpful citation was suggested by an anonymous reviewer.}, who state that ``[a]lmost all statistical problems can be considered as model selection problems'' (p. 3).
That monograph focuses on regression problems, and describes model selection based on hypothesis tests (including SSTs) as well as a range of methods that seek to optimize prediction errors, information criteria, or posterior probability for a range of different scenarios.
This paper is both more broad and much more restricted in its focus.
It is restricted to model selection based on hypothesis tests (more specifically, p-values), and in the particular case in which the researcher is choosing between a sequence of models, ranked by preference.
We are not presenting any new hypothesis tests (unlike, e.g. \citealt{vuong1989likelihood}) or any other way to decide between two competing models (e.g. \citealt{schwarz1978estimating})---instead, we assume that a hypothesis test already exists.
Our goal is to provide guidance on how to use p-values from that test to choose between models in a sequence.
On the other hand, the method we are introducing will apply in any scenario where such a hypothesis test is available.

\section{Background}\label{sec:setup}
\subsection{Sequences of Models and Tests}
Say, in specifying a model, a researcher must choose from a discrete set of specifications $\mathcal{S}_d$; $d=1,2,\dots,D$.
In this setup, we consider the set of $D$ candidate specifications as fixed, and not dependent on sample size.
The models are ordered by preference, subject to a testable assumption operationalized as $H_{0d}$; that is,
if $d>d'$, and $H_{0d}$ holds for both $\mathcal{S}_d$ and $\mathcal{S}_{d'}$,
then $\mathcal{S}_d$ is preferable to $\mathcal{S}_{d'}$.
The assumption $H_{0d}$ can have the same form for each $d$ (e.g. model $\mathcal{S}_d$ fits the data) or can be relative (e.g. model $\mathcal{S}_d$ fits the data as well as $\mathcal{S}_{d-1}$).
Denote the optimal specification choice as $\mathcal{S}_{\dstar}$, where
\begin{equation*}
  \dstar=\begin{cases}
  max\{d\in [1,D] : H_{0d}\text{ is true}\} &\text{ if }H_{0d}\text{ is true for some }d\in[1,D]\\%\text s.t.H_{0d} \text{ is true}\\
  0 & \text{otherwise}
  \end{cases}
\end{equation*}
Then we assume the following:
\begin{assumption}\label{ass:orderedH}
  If $d_2>d_1$, then $H_{0d_2}$ implies $H_{0d_1}$
\end{assumption}
That is, either $H_{0d}$ is false for all $d\in [1,D]$, in which case $\dstar=0$, $H_{0d}$ is true for all $d\in [1,D]$, in which case $\dstar=D$, or $H_{0d}$ is true for $1\le d\le \dstar$ and $H_{0d}$ is false for $\dstar<d \le D$.

Now assume the researcher has chosen a testing procedure for $H_{0d}$ producing a p-value for each $H_{0d}$, $p_1,\dots,p_D$.
Let $n_d$ be the sample size included in the test for $H_{0d}$.
This can take a number of forms---for instance, in the unemployment example of Figure \ref{fig:example}A, as $d$ increases the model becomes more parsimonious, but the data used to fit the model stays the same, so $n_d=n$ for all $d$.
In contrast, in the RDD example of Figure \ref{fig:example}B, a specification with larger $d$ includes more observations, so $d_2>d_1$ implies $n_2\ge n_1$.
When a dataset has a multilevel or hierarchical structure, the concept of a sample size becomes more subtle.
For instance, in a clustered survey sample $n_d$ may refer to the number of clusters included in a hypothesis test, rather than the number of observations.
$n_d$ plays no direct role in the computations we will introduce, so for the sake of generality we can afford to be slightly vague about its meaning.
When we wish to emphasize the dependence of $p_d$ on $n_d$, we write it as $p_{nd}$.
In any event, we assume that the test of each $H_{0d}$ is valid and asymptotically powerful.
That is, we assume:
\begin{assumption}\label{ass:pvals}
  For each $d=1,\dots,D$,
  \begin{itemize}
  \item If $d\le\dstar$, $p_{nd}\sim \mathcal{U}(0,1)$
  \item If $d>\dstar$, $p_{nd}\rightarrow_p 0$ as $n_d\rightarrow \infty$
  \end{itemize}
\end{assumption}
so that when $H_{0d}$ is true, the associated p-value is uniformly distributed, and when $H_{0d}$ is false the resulting p-value should be small in large samples.
In some cases, it will  suffice to substitute the following weaker form of Assumption \ref{ass:pvals}:
\begin{assumptionp}{\ref{ass:pvals}$'$}\label{ass:pvalsWeak}
  For each $d=1,\dots,D$, as $n_d\rightarrow \infty$,
\begin{itemize}
  \item If $H_{0d}$ is true, $\EE[p_{nd}]\rightarrow 1/2$
  \item If $H_{0d}$ is false, $p_{nd}\rightarrow_p 0$
  \end{itemize}
\end{assumptionp}
The goal here is to use $\bm{p}_D$ to choose a specification
$\hat{d}$ that is as large as possible without violating
the specification assumption encoded in $H_{0d}$.

For the methods we describe here and in the following section, it is
not necessary for the p-values $\bm{p}_D$ to be mutually independent;
indeed, they typically are not.

%% \subsection{Running Toy Example: Polynomial Regression}
%% To help fix ideas, we present a familiar toy example.
%% Say, for $i=1\dots n$, $X_i\sim \mathcal{U}(-1,1)$ and $Y_i=\sum_{j=0}^k\beta_j X_i^j+\epsilon_i$, where $\epsilon_i \mathop{\sim}\limits^{\mathrm{iid}}\mathcal{N}(0,\sigma^2)$ for some $\sigma^2>0$.
%% That is, random variable $Y$ is a polynomial function of $X$ of degree $k$, plus error.
%% A data analyst fitting a regression does not know $k$, and instead fits a sequence of model specifications

\subsection{Two Existing Approaches to Sequential Specification Tests}

In the contexts of sequential specification tests,
\citet{rao2001model} suggests a common approach to choosing a specification: for a pre-specified $\alpha \in (0,1)$, let
\begin{equation*}
  \dalphaU \equiv \begin{cases} \max\{d : p_d\ge\alpha\} & \text{if }\max{p_d}\ge\alpha\\
    0 & \text{otherwise}
    \end{cases}
\end{equation*}
That is, $\dalphaU$ is the largest value of $d$ for which $H_{0d}$
cannot be rejected at level $\alpha$ (and 0 if no such $d$ exists).
Although it may seem as though multiplicity corrections may be
necessary here, it turns out that this is not the case.
The ``stepwise intersection-union
principle'' \citep{berger1988, rosenbaum2008,hansen2015} insures that the
family-wise error rate is maintained, and the probability of falsely rejecting any null hypothesis $H_{0d}$ is bounded by $\alpha$:
\begin{prop}[\citealt{hansen2015}]
  Under Assumptions \ref{ass:orderedH} and \ref{ass:pvals}, $pr(\dalphaU<\dstar)\le \alpha$
\end{prop}
The proof can be found in \citet{hansen2015}.
Informally, note that (1) if $\dstar>0$, selecting $\dalphaU<\dstar$ entails rejecting $H_{0\dstar}$, (2) that $H_{0\dstar}$ is true, and (3) that
under Assumption 2, the probability of rejecting a true null at level $\alpha$ is equal to $\alpha$.
If $\dstar=0$, then $\dalphaU<\dstar$ is impossible by definition.
$\dalphaU$ is the specification that would result from testing null
hypotheses backwards: for $d'=D,D-1,\dots,d,\dots,1$, test $H_{0d'}$,
and stop testing at $d'=\dalphaU -1$, the first $d'$ for which
$p_{d'} \ge \alpha$.

Another common choice for $\hat{d}$  \citep[e.g.][]{lutkepohl2005new}
does not have this property.
Let
\begin{equation*}
\dalphaB\equiv min\{d: p_d<\alpha\}-1
\end{equation*}
$\dalphaB$ selects $\hat{d}$ to be the largest value of $d$ before the first
significant p-value ($\dalphaB=0$ if all p-values are $<\alpha$).
This is equivalent to the opposite procedure as $\dalphaU$: start with the $d'=1$
and test sequentially for larger values of $d'$ until the first
rejection, at $\dalphaB$, then stop; reject all null
hypotheses $H_{0d'}$ for $d'\ge \dalphaB$ and fail to
reject the rest.
This procedure does not control family-wise error rates, so it is likely
to reject more than $100\alpha$\% valid specifications.

Both $\dalphaU$ and $\dalphaB$ require the researcher to choose a rejection level $\alpha$ in advance, typically without much guidance or motivation.
Also, both procedures are susceptible to outlier p-values: an errant p-value exceeding $\alpha$ for a large value of $d$ will cause $\dalphaU$ to be too large, and an errant low p-value for a low value of $d$ will cause $\dalphaB$ to be too small.

%\citet{thall1997variable} suggested an intriguing, computationally-expensive alternative in the context of variable selection for regression modeling.

In the following section, we will suggest an alternative way to select a specification from a sequence of p-values that avoids these pitfalls by relying on the logic of estimation instead of the logic of hypothesis testing.

\section{The Change Point Estimator}

The asymptotic behavior described in Assumption \ref{ass:pvalsWeak}---when $d\le \dstar$, $\EE p_{dn}=1/2$ and when $d>\dstar$, $p_{dn}\rightarrow_p 0$---suggests a least-squares estimator for $\dstar$:
\begin{equation}\label{eq:dhatm}
\dhatm\equiv \displaystyle\argmin_{d} \displaystyle\sum_{t\le d} (p_t -1/2)^2 +
\displaystyle\sum_{t>d} p_t^2=\displaystyle\argmax_d \displaystyle\sum_{t\le d} (p_t-1/4).
\end{equation}
In other words, the estimate $\dhatm$ is the point at which the
p-values cease behaving as p-values testing a true null, with mean
$1/2$, and instead are drawn from a distribution with a lower mean.

Because p-values are inherently unpredictable when the null hypothesis is true, even in large samples $\dhatm$ may choose a sub-optimal specification---i.e. $pr(\dhatm<\dstar)>0$ for all $n$.
In particular, note that under Assumption \ref{ass:pvals}, $pr(p_{\dstar}<1/4)=1/4$ regardless of $n_d$, and that when $p_{\dstar}-1/4<0$, $\dhatm \neq \dstar$, because
$\sum_{d\le\dstar-1} (p_d-1/4)>\sum_{d\le \dstar} (p_d-1/4)$.
However, $\dhatm$ is asymptotically conservative:
\begin{prop}\label{prop:conservative}
Under Assumptions \ref{ass:orderedH} and \ref{ass:pvalsWeak},
$pr(\dhatm>\dstar)\rightarrow 0$ as $\displaystyle\min_{1\le d\le D}\{n_d\}\rightarrow \infty$.
\end{prop}
\begin{proof}
For each $d>\dstar$, $pr(p_d -1/4>0)\rightarrow 0$, implying that for all $d'$, $pr(\sum_{\dstar <t\le d'}
(p_t-1/4)>0)\rightarrow 0$.
Therefore, for $\dstar<d\le D$, $pr(\sum_{t\le d} (p_t-1/4)> \sum_{t\le
  \dstar} (p_t-1/4))\rightarrow 0$.
\end{proof}
If $\dhatm>\dstar$, then specification $\mathcal{S}_{\dhatm}$
violates the assumption encoded in $H_{0d}$;
as sample size increases, the probability of this event decreases to
zero.
The same property holds for $\dalphaU$, with $\alpha>0$ fixed, for
the same reason.

In a way, $\dhatm$ is similar to $\dhatU_{0.25}$, the largest $d$ for
which $p_d>\alpha=0.25$, because both penalize p-values lower than
$0.25$.
However, they are not equivalent, as the following proposition shows:
\begin{prop}\label{prop:d25}
$\dhatm \le \dhatU_{0.25}$, with $pr(\dhatm < \dhatU_{0.25})>0$.
\end{prop}
\begin{proof}
By definition, $p_d<0.25$ for all $d>\dhatU_{0.25}$. Therefore,
$\sum_{t=\dhatU_{0.25}+1}^{d'}(p_t-1/4)<0$ for all $d'\ge
\dhatU_{0.25}+1$, which in turn implies that $\sum_{t\le
  \dhatU_{0.25}}(p_t-1/4)>\sum_{t\le d'}(p_t-1/4)$, proving that
$\dhatm\le \dhatU_{0.25}$. On the other hand, if, say,
$p_{\dhatU_{0.25}-1}+p_{\dhatU_{0.25}}<1/2$, or, more generally,
$\sum_{t=d'}^{\dhatU_{0.25}}(p_t-1/4)<0$, then $\dhatm<\dhatU_{0.25}$.
\end{proof}

In general, the difference between $\dalphaU$ and $\dhatm$ will be
most pronounced when the distributions of p-values for $d>\dstar$ are
not monotonically decreasing in probability. In such a scenario, it
is most probable that an errant p-value for $d>>\dstar$ will be
greater than $\alpha$; one p-value determines $\dalphaU$, but
$\dhatm$ relies on the entire set of p-values.

\subsection{A More Flexible $\dhatm$}
In finite samples, p-values from tests of false null hypotheses will typically be greater than zero.
Similarly, many hypothesis tests are asymptotic and may not yield
uniformly-distributed p-values in finite samples.
Still, p-values from sequential specification tests may exhibit something similar to the
dichotomous behavior that motivates $\dhatm$, in which p-values for
$d\le \dstar$ are distributed differently than p-values for
$d>\dstar$.
For this reason, \citet{mallik} suggested a more flexible estimate:
\begin{equation*}
  \dhatmab \equiv \displaystyle\argmin_{d; 0<b<a<1}
  \displaystyle\sum_{t\le d} ( p_t
  -a)^2+\displaystyle\sum_{t>d} (p_t-b)^2
\end{equation*}
Like $\dhatm$, model selector $\dhatmab$ looks for behavior that
differs between p-values testing true and false null hypotheses.
Unlike $\dhatm$, it does not depend on theoretically established
distributions for these p-values, but searches over a grid for their
location parameters.
$\dhatmab$ will be more computationally expensive to compute than
$\dhatm$, but will may yield better results, especially in small
samples.

%\section{Results}

\section{A Simulation Study}\label{sec:simulation}

This section will present a simulation study to compare the
behavior of model selectors $\dalphaU$, $\dalphaB$, and
$\dhatm$.
The simulation imagines a sequence of 10 models, ordered from
least to most preferable.
The first 5 models are well specified; thereafter the models
are increasingly misspecified following a linear gradual change model \cite[c.f.][]{vogt2015detecting,shao2016detection}.
Each model is assessed with a $Z$-test.
For models $d=1,\dots,$5, the test statistic
$Z_d\sim\mathcal{N}(0,1)$, the standard normal distribution.
For models $d=$6,$\dots$,10, the test statistic is
distributed as $Z\sim\mathcal{N}\big\{\beta(d-5),1\big\}$, where
the slope parameter $\beta$ controls the power of specification tests for
these misspecified models,
which increases with $d$ for values of $d>$5.
Specification p-values are generated by comparing all of these
simulated test statistics against the null distribution
$\mathcal{N}(0,1)$.

\begin{figure}

\includegraphics[width=\maxwidth]{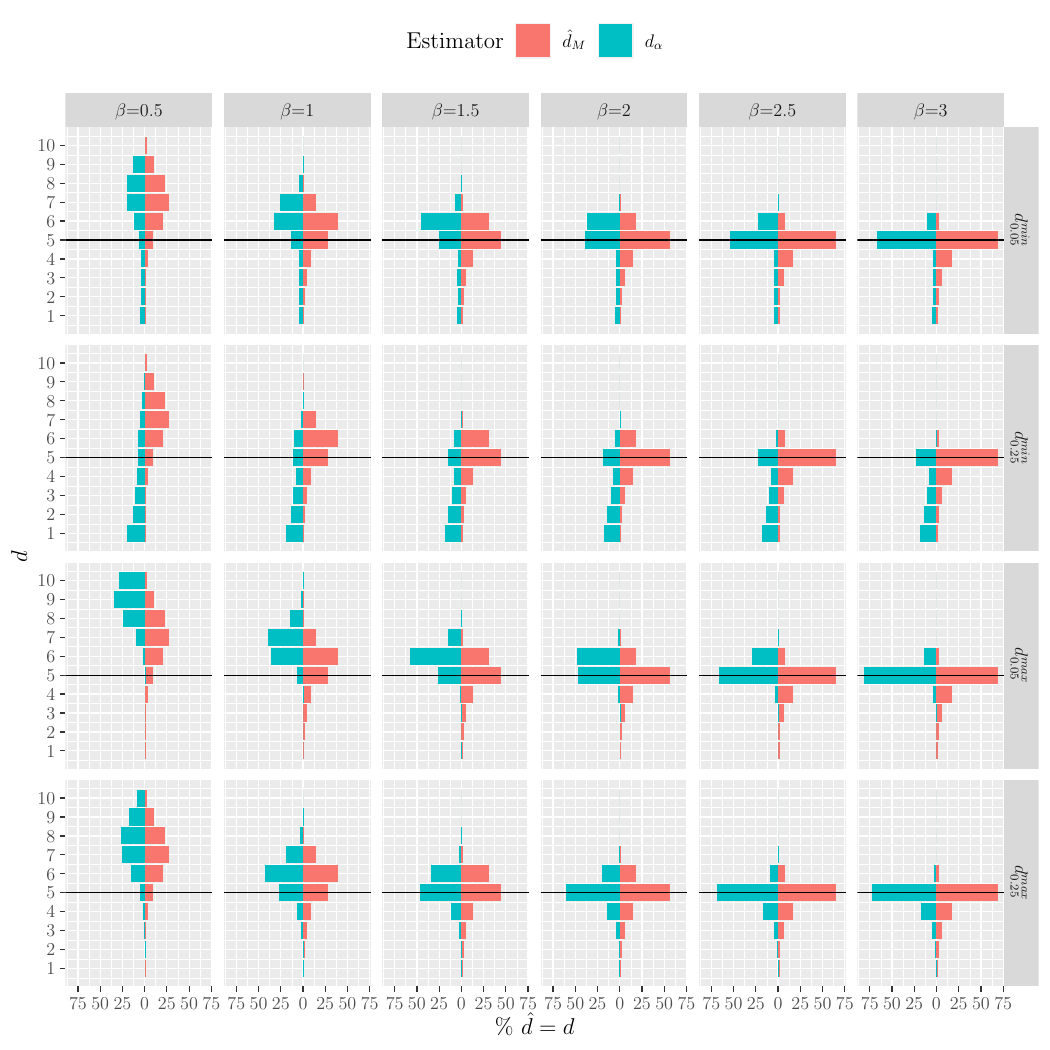}

\caption{Results from  $10^{4}$  simulation runs comparing $\dhatm$ to $\dhatU$
  and $\dhatB$ with $\alpha=0.05$ and $0.25$. Each row compares either
  $\dhatU$ or $\dhatB$ to the same set of $\dhatm$ estimates. Each bar
  represents the percent of runs in which an estimator selects each
  possible model, indexed as $d=1,\dots$,10. Model
  $d=$5 (indicated with a horizontal line) is the optimal
  model, with models $d>$5 misspecified, and models
  $d<$5 well-specified but suboptimal. }
\label{fig:simulation}
\end{figure}

\begin{table}

\begin{tabular}[t]{lrrrrrrrrr}
\toprule
\multicolumn{1}{c}{ } & \multicolumn{3}{c}{$\beta=0.5$} & \multicolumn{3}{c}{$\beta=2$} & \multicolumn{3}{c}{$\beta=3$} \\
\cmidrule(l{3pt}r{3pt}){2-4} \cmidrule(l{3pt}r{3pt}){5-7} \cmidrule(l{3pt}r{3pt}){8-10}
Estimator & RMSE & \%Opt. & \%$>d^*$ & RMSE & \%Opt. & \%$>d^*$ & RMSE & \%Opt. & \%$>d^*$\\
\midrule
$\dhatU_{0.05}$ & 15.3 & 0 & 100 & 0.6 & 47 & 50 & 0.2 & 82 & 14\\
$\dhatU_{0.25}$ & 8.5 & 5 & 93 & 0.6 & 60 & 20 & 0.6 & 73 & 3\\
$\dhatB_{0.05}$ & 6.3 & 6 & 65 & 1.8 & 39 & 37 & 1.4 & 67 & 11\\
$\dhatB_{0.25}$ & 5.4 & 7 & 16 & 4.7 & 19 & 5 & 4.7 & 23 & 1\\
$\dhatm$ & 6.1 & 9 & 84 & 1.1 & 56 & 18 & 1.0 & 69 & 3\\
\bottomrule
\end{tabular}

\caption{Some results from $10^{4}$ simulation runs comparing $\dhatm$ to $\dhatU$
  and $\dhatB$ with $\alpha=0.05$ and $0.25$. For $\beta=0.5,2,3$, the
  root-mean-squared error (RMSE) of each estimator
  $\left\{\overline{(\hat{d}-d^*)^2}\right\}^{1/2}$ and the percentages each
  estimator chose the optimal model (\%Opt.) or chose a misspecified
  model (\%$>d^*$)}
\label{tab:simulation}
\end{table}

Figure \ref{fig:simulation} and Table \ref{tab:simulation} give the results of the simulation study,
comparing $\dhatm$ to $\dhatB_{0.05}$, $\dhatB_{0.25}$, $\dhatU_{0.05}$
and $\dhatU_{0.25}$, respectively.
Table \ref{tab:simulation} compares all five model selectors at
$\beta=0.5$, 2, and 3 on three criterion: root mean-squared-error
($RMSE(x)=\left\{\overline{(x-d^*)^2}\right\}^{1/2}$), a measure of how close, in
general, the estimator is to the optimal value, the percentage
of runs in which it chose the optimal value $d^*$ (\%Opt.) and the
percentage of runs in which it chose a misspecified model, i.e. chose
$d>d^*$ (\%$>d^*$).

In Figure \ref{fig:simulation}, each bar represents the percentage of the times each model selector chose
model $d$, with $d=1,\dots$,10.
Model $d=$5 (indicated with a horizontal line) is the optimal
model, with models $d>$5 misspecified, and models
$d<$5 well-specified but suboptimal.
Each column of Figure \ref{fig:simulation} corresponds to a different
value for the slope parameter $\beta\in\{0.5,1,1.5,2,2.5,3\}$.
As $\beta$ increases, so does the power of the specification test,
allowing the test to reject misspecified models at smaller values of
$d>d^*$.
Each row compares the same set of $\dhatm$ to either $\dhatB_{0.05}$,
$\dhatB_{0.25}$, $\dhatU_{0.05}$, or $\dhatU_{0.25}$.

When $\beta=0.5$, the power to detect misspecification for models $d>d^*$
is relatively low.
$\dhatm$, $\dhatB_{0.05}$, and both $\dhatU$ model selectors tend to
choose models that are too big.
That said, of those four estimators, $\dhatm$ has the smallest root
mean-squared-error (RMSE) and
$\dhatm$ is most likely of all model selectors to choose the optimal
model $d^*$.
$\dhatB_{0.25}$, which is the least likely to
recommend a misspecified model, tends to recommend $d=1$, the smallest
possible, least-optimal model.

As $\beta$ increases, the performance of all five model selectors
improves.
Throughout, $\dhatm$ is competitive in all three criteria and,
arguably, balances them the best.
At $\beta=2$, and $\beta=3$, the $\dhatU$ estimators have better RMSE and tend to pick
the optimal model slightly more often than $\dhatm$, but are more
likely to pick misspecified models.
$\dhatB_{0.25}$ is the least likely to pick misspecified models, but
the models it does pick tend to be much too small.

In general, $\dhatm$ tends to be more conservative than $\dhatU$ or
$\dhatB$ with $\alpha=0.05$ but much less conservative than
$\dhatB_{0.25}$.
Its performance is most similar to $\dhatU_{0.25}$, while being
slightly more conservative.

The appendix gives a larger version of Table \ref{tab:simulation}
including results from when there are $D=20$ candidate models or
$D=10$, as here, and when the optimal model $d^*=2$, $d^*=5$ (as
here), and, when $D=20$, $d^*=10$.
Broadly speaking, the patterns of performance are similar as $D$ and
$d^*$ vary.

\section{Two Data Examples}\label{sec:examples}

\subsection{Lag Order in Autoregression Models: US Total Unemployment}\label{sec:unemployment}

Figure \ref{fig:example}B shows the natural logarithm of the United States total
unemployment rate from 1890 to 2016.
The data were combined from the ``Nelson \& Plosser extended data
set'' provided in the \texttt{urca} library in \texttt{R}
\citep{urca,Rcite}, which covers years 1890--1988, and a downloadable
dataset from the United States Bureau of Labor Statistics, itself
derived from the Current Population Survey, which covers years
1947--2015 \citep{cps}.
The two datasets agree on the overlapping years.

Assume that the time series follows an ``AR($d$)'' model; that is,
\begin{equation}\label{eq:arp}
unemp_t=\mu + \displaystyle\sum_{i=1}^d \phi_i unemp_{t-i}+\epsilon_t
\end{equation}
where $\mu$ and $\{\phi_i\}_{i=1}^d$ are parameters to be estimated
and $\epsilon_t$ is white noise.
In this model, the unemployment in one year is a function of
unemployment rates in the previous $d$ years, but conditionally
independent of even earlier measurements.

Having settled on model (\ref{eq:arp}), the analyst must choose $d$,
the lag order.
Sequential specification tests can be useful here \citep[e.g.][]{practitionersGuide}.
Consider the null hypothesis $H_d: \phi_i=0$ for all $i>d$;
a researcher could test a sequence of such null hypotheses, for a set
of plausible values of $d$, and choose the $d$ based on the results.
Other options for choosing $d$ include optimizing
information criteria
\citep{akaike1969fitting,schwarz1978estimating}.
For instance, choosing the model that minimizes \textsc{aic}, defined as
$2(d+2) -2log(\hat{L}_d)$, where $\hat{L}_d$ is the maximized likelihood
of the $AR(d)$ model, or \textsc{bic}, which is defined as $log(n)(d+2) -2log(\hat{L}_d)$.
%\citet{potscher1991effects} points out that differences in \textsc{aic} or
%\textsc{bic} are essentially likelihood ratio test statistics.
%% Sequential specification tests can assist a modeler to
%% choose the smallest model that is still approximately correct, as
%% opposed to the model that maximizes predictive accuracy as measured
%% by, say, mean squared error.
A large literature surrounds this important question \citep[See,
e.g.][and the citations
therein]{mcquarrie1998regression,liew2004lag}. This section is not
meant as a complete treatment, or even an overview, of lag order
selection, but as an illustration of sequential specification tests in a well-known area.

Figure \ref{fig:tspvalues1} gives the p-values from a sequence
likelihood ratio tests, as described in \citet[][Ch.1]{urca}, which
discussed a similar dataset.
For each candidate lag order $d$, the likelihood ratio test compares
twice the log of the ratio of the likelihoods of $AR(d+1)$ and $AR(d)$ models
to a $\chi^2_1$ distribution.
If the $AR(d+1)$ model fits much better than the $AR(d)$ model, a lag
order of $d$ may not be sufficient.
The p-values follow a stark pattern: for $d<5$, they are close to
zero, while for $d\ge5$, they appear roughly uniformly distributed.

Table \ref{tab:ts}, and vertical lines in Figure \ref{fig:tspvalues1},
show the lag order choices from $\dalphaU$, $\dalphaB$, $\dhatm$, and
$\dhatmab$, which are based on the p-values, and the lag orders that
minimize \textsc{aic} and \textsc{bic}, based directly on the models' likelihood and numbers
of parameters.
Here, smaller models are preferable to larger models, so $\dstar$ is
the smallest acceptable value for $d$.
%This is the opposite of the regression discontinuity case, which
%attempted to find the largest dataset on which to fit the model.

The change-point selectors $\dhatm$ and $\dhatmab$ both selected a
lag order of  5, consistent with the casual
observation that p-values for lags less than this value are very
small, while those greater appear approximately uniform.
Incidentally, the two information criteria considered, \textsc{aic} and \textsc{bic},
agreed with this choice, as did $\dhatU_{0.25}$.
In contrast, $\dhatU_{0.05}$ chose a smaller lag order of
3, because the corresponding p-value of
0.066 slightly exceeds the threshold of 0.05.

At the other extreme, the $\dalphaB$ selectors both chose very large
models with
$d=$17 and 19, due to the presence of
of small p-values of
0.044 and
0.198 at
$d=$16 and 18.

This example illustrates how considering the entire distribution of
p-values, as $\dhatm$ does, can lead to better model selection than
considering only the small (as in $\dalphaB$) or large ($\dalphaU$)
values.

\begin{figure}

\includegraphics[width=\maxwidth]{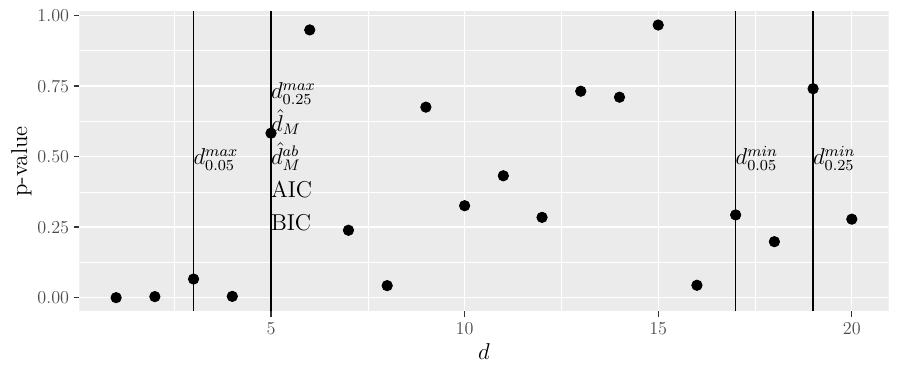}

\caption{P-values from likelihood-ratio tests of model fit,
  comparing models AR($d$) with AR($d+1$) in the annual total US
  unemployment rate (logged) time series.}
\label{fig:tspvalues1}
\end{figure}

% latex table generated in R 4.0.3 by xtable 1.8-4 package
% Tue Mar 09 14:48:38 2021
%% \begin{table}[ht]
%% \centering
%% \begin{tabular}{ccccccccc}
%%   \hline
%%  & $\dhatU_{0.05}$ & $\dhatU_{0.25}$ & $\dhatB_{0.05}$ & $\dhatB_{0.25}$ & $\dhatm$ & $\dhatmab$ & AIC & BIC \\
%%   \hline
%% Lag Order &   3 &   5 &  17 &  19 &   5 &   5 &   5 &   5 \\
%%    \hline
%% \end{tabular}
%% \caption{Lag order selections for an $AR(d)$ model of the US unemployment time series.}
%% \label{tab:ts}
%% \end{table}

% latex table generated in R 4.1.1 by xtable 1.8-4 package
% Thu Nov 17 11:42:00 2022
\begin{table}[ht]
\centering
\begin{tabular}{ccccccccc}
  \hline
 & $\dhatU_{0.05}$ & $\dhatU_{0.25}$ & $\dhatB_{0.05}$ & $\dhatB_{0.25}$ & $\dhatm$ & $\dhatmab$ & AIC & BIC \\
  \hline
Lag Order &   3 &   5 &  17 &  19 &   5 &   5 &   5 &   5 \\
   \hline
\end{tabular}
\caption{Lag order selections for an $AR(d)$ model of the US unemployment time series.}
\label{tab:ts}
\end{table}

\subsection{Sequential specification tests in Regression Discontinuity Bandwidth Selection:
  Estimating the Effect of  Academic Probation on College GPAs}\label{sec:rdd}

At many universities, students who fail to achieve a minimum GPA $c$ are
put on academic probation.
\citet{lso} recognized that academic probation can form a regression discontinuity
design, in which treatment is a function of a ``running
variable'' with a pre-determined cutoff.
Specifically, probation $Z$ is a function of
a ``running variable'' $R$, students' GPAs: students with $R<c$ are
put on probation---$Z=1$---and students with $R>c$ are not, $Z=0$.
That being the case, students with GPAs just below $c$ may be
comparable to students with GPAs just above $c$, so comparing these two
sets of students allows researchers to estimate the effect of probation on
outcomes $Y$ (perhaps after adjusting for $Y$'s relationship with $R$).
The challenge becomes defining ``just above'' and ``just below''---that
is, selecting a ``bandwidth'' $\omega^*>0$ such that subjects $i$ with $R_i\in
(c-\omega^*,c)$ are suitably comparable to subjects with $R_i\in(c,c+\omega^*)$.

A number of authors \citep[e.g.]{lee,cft,mattai}
recommend sequential specification tests, using baseline covariates
$X$, as part of the procedure for choosing $\omega$.
At a sequence of candidate bandwidths $0<\omega_1<\dots<\omega_d<\dots<\omega_D$,
they recommend testing the equality of covariate means (again, perhaps
after adjusting for $R$) between subjects with $R_i\in
(c-\omega_d,c)$ and those with $R_i\in(c,c+\omega_d)$, and choosing a
bandwidth $\omega^*=\omega_{d^*}$.
These are essentially placebo tests---since the treatment cannot
affect baseline covariates, differences in covariate means between
treated and untreated subjects must be an indicator of incomparability
between the groups, or model misspecification.

In a secondary analysis of the academic probation dataset, \citet{lrd}
chose an RDD bandwidth using a set of seven baseline
covariates: students' high-school GPA (expressed in percentiles), age
at college matriculation, number of attempted credits, gender, native
language (English or other), birth place (North America or other) and
university campus (the university consisted of three campuses).
For each covariate $X_k$ and for each candidate bandwidth $\omega_d$, they
let $p_{kd}$ be the p-value corresponding the coefficient on $Z$ from
the regression of $X_k$ on $R$ and $Z$, fit to the subset of students
with $R\in (c-\omega_d,c+\omega_d)$.
These regression models were linear for continuous covariates and
logistic for binary covariates, with heteroskedasticity-consistent sandwich standard errors \citep{sandwich1,sandwich2,sandwich3}.
Then, the omnibus specification p-value for bandwidth $\omega_d$ was
$p_d= min\{1,7p_{1d},\dots,7p_{7d}\}$, the minimum of the
Bonferroni-adjusted p-values $p_{kd}$.

\begin{figure}
\begin{knitrout}
\definecolor{shadecolor}{rgb}{0.969, 0.969, 0.969}\color{fgcolor}
\includegraphics[width=\maxwidth]{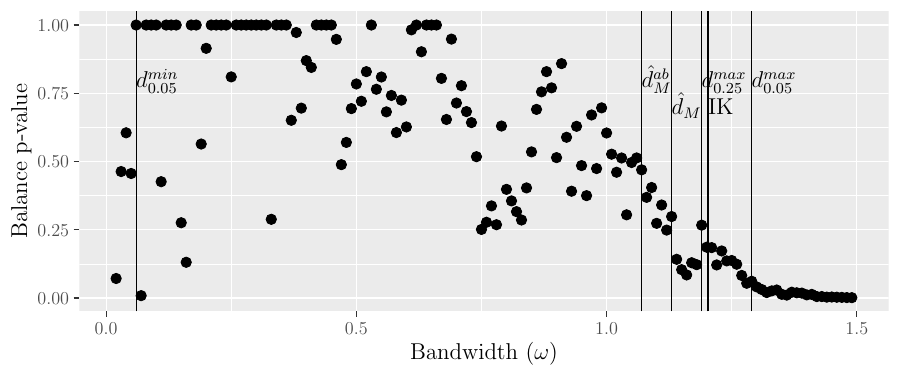}
\end{knitrout}
\caption{P-values for balance in all seven covariates
  from the \citet{lso} analysis, following the method in \citet{lrd}. Vertical lines denote bandwidth choices using
  different criteria.}
\label{fig:rdpvalues1}
\end{figure}

The resulting p-values are plotted in  Figure \ref{fig:rdpvalues1},
with bandwidth selections corresponding to
$\dhatU_{0.05}$, $\dhatU_{0.25}$, $\dhatB_{0.05}$, $\dhatm$, and $\dhatmab$.
Also plotted is the more conventional bandwidth recommended by
\citet{IK}, denoted IK,
which is based on non-parametric estimates of the curvature of
the regression function of $Y$ on $R$, rather than covariate placebo
tests.
These bandwidth selections are also listed in Table \ref{tab:RDD}.

For most small bandwidths $\omega_d$, $p_d$ is fairly large, and in many
cases equal to 1.
This apparent super-uniform distribution is probably due to the
conservative Bonferroni correction applied to the p-values from
individual covariates.
On the other hand, at the smallest candidate bandwidth
$\omega_1=$0.02, the p-value is $p_1=$%
0.072, and the p-value at the
6th bandwidth, $\omega_{6}=0.07$,
another small bandwidth, is
$p_{6}=0.009$.
After around $\omega=0.75$, the p-values begin decreasing, until by
$\omega=1.5$, the p-values are all close to zero.

% latex table generated in R 4.1.1 by xtable 1.8-4 package
% Thu Nov 17 18:43:41 2022
\begin{table}[ht]
\centering
\begin{tabular}{rccc}
  \hline
 & $\hat{d}$ & Bandwidth & Effect (95\% CI) \\
  \hline
$\dhatU_{0.05}$ & 128 & 1.29 & 0.22 (0.18,0.26) \\
  $\dhatU_{0.25}$ & 118 & 1.19 & 0.23 (0.18,0.27) \\
  $\dhatB_{0.05}$ & 5 & 0.06 & 0.1 (-0.17,0.37) \\
  $\dhatB_{0.25}$ &  & N/A & N/A \\
  $\dhatm$ & 112 & 1.13 & 0.23 (0.19,0.27) \\
  $\dhatmab$ & 106 & 1.07 & 0.22 (0.18,0.27) \\
  IK & 119 & 1.2 & 0.23 (0.19,0.28) \\
   \hline
\end{tabular}
\caption{Selected regression discontinuity bandwidths (``$\hat{d}$'' is the point in the sequence selected, and ``Bandwidth'' is the actual bandwidth) using covariate balance tests, or using the method described in \cite{IK}, along with their associated estimates for the average treatment effect of academic probation on subsequent GPAs (ATE), with 95\% confidence intervals in parentheses.}
\label{tab:RDD}
\end{table}

Since the p-value at the smallest candidate bandwidth,
$p_1=0.072<0.25$, the model selector
$\dhatB_{0.25}$ does not select anything---there is no $d'$ small
enough so that $p_{d}<0.25$ for all $d\le d'$.
Similarly, the very low p-value at the 6th bandwidth
causes $\dhatB_{0.05}$ to select a relatively small bandwidth of 0.07.
This illustrates the sensitivity of $\dhatB$ to outlier p-values at
small $d$.

The remaining selectors all recommend bandwidths greater than 1,
ranging from
$\dhatmab$, which recommends bandwidth $\omega_{\dhatmab}$=1.07 to
$\dhatU_{0.05}$, which recommends bandwidth
$\omega_{\dhatU_{0.05}}=$1.29.
As in the simulation and the unemployment example, $\dhatm$ and
$\dhatU_{0.25}$ are quite close to each other.
The similarity of the IK bandwidth of $1.2$
to the bandwidths selected by $\dhatU$ and $\dhatm$ suggests an
encouraging agreement, in this example, between covariate-based
bandwidth selection and the more conventional RDD approach.

It is worth noting that super-uniformity of the p-values for small
bandwidths inflates the sums $\sum_{t\le d} (p_t-1/4)$ in
\eqref{eq:dhatm}.
Therefore, in this case $\dhatmab$, which relies less on the uniform
model for p-values under $H_0$, may be a more appropriate choice
than $\dhatm$.

Table \ref{tab:RDD} also lists estimated treatment effects of academic
probation on students' subsequent GPAs, along with 95\% confidence
intervals.
The effects were estimated following the method described in
\citet{lrd}, with the exception of the estimate for the IK bandwidth
which used local linear regression, as implemented in the \texttt{R}
package \texttt{rdd} \citep{rdd}.
With the exception of the effect corresponding to $\dhatB_{0.05}$, all
estimated effects are roughly equal, slightly less than 1/4 of a grade point.
Actually, the confidence interval corresponding to the $\dhatB_{0.05}$
bandwidth,
(-0.17,0.37)
is wide enough to contain both a negative academic probation effect
along with all of the other estimated effects and confidence intervals.
The conservativism of $\dhatB_{0.05}$ prevents efficient effect
estimation; the conservativism of $\dhatB_{0.25}$ prevents
estimation altogether.

\section{Discussion}\label{sec:discussion}

As long as data analysts use specification tests and p-values to
select their models, decision rules translating a sequence of p-values to a
model choice will be necessary.
Currently, the most common approach compares the p-values to a
pre-specified threshold.
This approach turns the logic of null
hypothesis testing on its head, using p-values to identify
well-specified models---i.e. true null hypotheses---rather than to
reject misspecified models.
Moreover, the $\dhatB_{\alpha}$ approach, by failing to control the
familywise type-I error rate, can be extremely conservative,
for instance recommending very high lag order in the unemployment
example of Section \ref{sec:unemployment} and very small bandwidths
(or none at all) in the academic probation example of Section \ref{sec:rdd}.
In contrast, $\dhatU_{\alpha}$ does control familywise type-I error rates.
However, both threshold-based approaches, $\dhatB$ and $\dhatU$,
require specifying a threshold, and there is rarely any clear
guidance on how to do so.

The alternatives introduced here, $\dhatm$ and $\dhatmab$, drawn from
the change-point literature, skirt these issues entirely.
Rather than using p-values to reject (or accept) null hypotheses, they
examine the full distribution of p-values.
They require no arbitrary threshold to be specified.
As shown in the simulation study and the two data examples, they tend
to avoid the conservativism and outlier sensitivity of $\dhatB_{0.25}$
and the anti-conservativism of $\dhatU_{0.05}$ (indeed, Proposition
\ref{prop:conservative} states that $\dhatm$ is asymptotically conservative).

Actually, the simulation results and examples show that $\dhatm$ tends
to agree with $\dhatU_{0.25}$, which itself performs rather well.
This suggests that $\dhatU$ with a default threshold of 0.25 may be a
good option for data analysts who wish to continue using
threshold-based approaches.

There are several open questions regarding $\dhatm$'s behavior and
use.
First, it is unclear whether or when the more flexible version
$\dhatmab$ should be preferred to $\dhatm$; there is good reason to
expect it to perform better when sample sizes are small, but is there
a cost associated with using $\dhatmab$ in larger samples?
Further, there may be ways to construct sequential specification tests
in a way that improves $\dhatm$'s performance.
How to best construct specification tests for different model
selectors is a topic for future research.

Ultimately, the goal of model selection is to produce parameter
estimates or predictions with desired properties.
Ideally, researchers would select a model with this end in mind;
however, the effect of model selection on final estimates or
predictions depends heavily on specific circumstances.
That said, a careful study of the effect of $\dhatm$ on estimates and
predictions in a wide range of cases could be useful.

Model selectors $\dhatm$ and $\dhatmab$ may be particularly useful in
measurement modeling, where model choice based on sequences of
p-values are common (such as to select the number of components in factor or
latent class analysis).
The encouraging performance of $\dhatm$ when $d^*=2$ suggests that
$\dhatm$ may be appropriate even when the true number of components is
small.

\bibliographystyle{apalike}
\bibliography{sst}

%\end{document}
\clearpage

\appendix
\section{Simulation Results with different $D$ and $d^*$}

Table \ref{tab:fullSimulation} gives the results of Table
\ref{tab:simulation} for varying values of $D$ (the total number of
candidate models) and $d^*$ (the optimal model).

\begin{table}[!h]

\begin{tabular}[t]{rrlrrrrrrrrr}
\toprule
\multicolumn{3}{c}{ } & \multicolumn{3}{c}{$\beta=0.5$} & \multicolumn{3}{c}{$\beta=2$} & \multicolumn{3}{c}{$\beta=3$} \\
\cmidrule(l{3pt}r{3pt}){4-6} \cmidrule(l{3pt}r{3pt}){7-9} \cmidrule(l{3pt}r{3pt}){10-12}
\rotatebox{45}{$D$} & \rotatebox{45}{$d^*$} & \rotatebox{45}{Est.} & \rotatebox{45}{RMSE} & \rotatebox{45}{\%Opt.} & \rotatebox{45}{\%$>d^*$} & \rotatebox{45}{RMSE} & \rotatebox{45}{\%Opt.} & \rotatebox{45}{\%$>d^*$} & \rotatebox{45}{RMSE} & \rotatebox{45}{\%Opt.} & \rotatebox{45}{\%$>d^*$}\\
\midrule
 &  & $\dhatU_{0.05}$ & 21 & 0.1 & 100 & 1 & 48.6 & 49 & 0 & 81 & 14.7\\
\cmidrule{3-12}
 &  & $\dhatU_{0.25}$ & 10 & 4.5 & 94 & 0 & 60.4 & 19 & 0 & 73 & 3.1\\
\cmidrule{3-12}
 &  & $\dhatB_{0.05}$ & 8 & 7.2 & 83 & 1 & 47.6 & 43 & 0 & 77 & 13.2\\
\cmidrule{3-12}
 &  & $\dhatB_{0.25}$ & 2 & 17.7 & 39 & 0 & 46.4 & 11 & 0 & 55 & 1.8\\
\cmidrule{3-12}
 & \multirow{-5}{*}{\raggedleft\arraybackslash 2} & $\dhatm$ & 6 & 9.8 & 85 & 0 & 60.5 & 18 & 0 & 73 & 2.7\\
\cmidrule{2-12}
 &  & $\dhatU_{0.05}$ & 15 & 0.2 & 100 & 1 & 47.3 & 50 & 0 & 82 & 14.2\\
\cmidrule{3-12}
 &  & $\dhatU_{0.25}$ & 9 & 4.9 & 93 & 1 & 60.0 & 20 & 1 & 73 & 3.0\\
\cmidrule{3-12}
 &  & $\dhatB_{0.05}$ & 6 & 6.0 & 65 & 2 & 39.2 & 37 & 1 & 67 & 11.1\\
\cmidrule{3-12}
 &  & $\dhatB_{0.25}$ & 5 & 7.1 & 16 & 5 & 19.1 & 5 & 5 & 23 & 0.7\\
\cmidrule{3-12}
\multirow{-10}{*}{\raggedleft\arraybackslash 10} & \multirow{-5}{*}{\raggedleft\arraybackslash 5} & $\dhatm$ & 6 & 8.9 & 84 & 1 & 56.4 & 18 & 1 & 69 & 2.7\\
\cmidrule{1-12}
 &  & $\dhatU_{0.05}$ & 22 & 0.1 & 100 & 1 & 48.9 & 48 & 0 & 81 & 15.0\\
\cmidrule{3-12}
 &  & $\dhatU_{0.25}$ & 10 & 5.1 & 93 & 0 & 61.0 & 19 & 0 & 72 & 3.5\\
\cmidrule{3-12}
 &  & $\dhatB_{0.05}$ & 7 & 7.3 & 83 & 1 & 47.4 & 43 & 0 & 77 & 13.5\\
\cmidrule{3-12}
 &  & $\dhatB_{0.25}$ & 2 & 17.5 & 39 & 0 & 45.8 & 11 & 0 & 54 & 1.9\\
\cmidrule{3-12}
 & \multirow{-5}{*}{\raggedleft\arraybackslash 2} & $\dhatm$ & 6 & 10.0 & 84 & 0 & 61.1 & 18 & 0 & 72 & 3.1\\
\cmidrule{2-12}
 &  & $\dhatU_{0.05}$ & 22 & 0.1 & 100 & 1 & 47.8 & 50 & 0 & 81 & 14.5\\
\cmidrule{3-12}
 &  & $\dhatU_{0.25}$ & 10 & 4.6 & 94 & 1 & 60.1 & 20 & 1 & 73 & 3.2\\
\cmidrule{3-12}
 &  & $\dhatB_{0.05}$ & 8 & 6.1 & 72 & 2 & 39.3 & 38 & 1 & 67 & 11.4\\
\cmidrule{3-12}
 &  & $\dhatB_{0.25}$ & 5 & 7.3 & 17 & 5 & 19.2 & 5 & 5 & 23 & 0.8\\
\cmidrule{3-12}
 & \multirow{-5}{*}{\raggedleft\arraybackslash 5} & $\dhatm$ & 6 & 8.9 & 84 & 1 & 56.4 & 18 & 1 & 69 & 2.6\\
\cmidrule{2-12}
 &  & $\dhatU_{0.05}$ & 22 & 0.1 & 100 & 1 & 47.8 & 50 & 0 & 81 & 14.8\\
\cmidrule{3-12}
 &  & $\dhatU_{0.25}$ & 10 & 4.8 & 94 & 1 & 60.1 & 20 & 1 & 73 & 2.9\\
\cmidrule{3-12}
 &  & $\dhatB_{0.05}$ & 18 & 4.7 & 54 & 13 & 31.1 & 29 & 12 & 51 & 8.6\\
\cmidrule{3-12}
 &  & $\dhatB_{0.25}$ & 35 & 1.9 & 4 & 35 & 4.5 & 1 & 34 & 6 & 0.2\\
\cmidrule{3-12}
\multirow{-15}{*}{\raggedleft\arraybackslash 20} & \multirow{-5}{*}{\raggedleft\arraybackslash 10} & $\dhatm$ & 6 & 9.4 & 84 & 2 & 56.4 & 17 & 2 & 68 & 2.5\\
\bottomrule
\end{tabular}

\caption{Simulation runs with the total number of models to compare
  $D\in\{10,20\}$ and the optimal model $d^*\in\{2,5,10\}$ comparing $\dhatm$ to $\dhatU$
  and $\dhatB$ with $\alpha=0.05$ and $0.25$. For $b=0.5,2,3$, the
  root-mean-squared error (RMSE) of each estimator
  $\overline{(\hat{d}-d^*)^2}^{0.5}$ and the percentages each
  estimator chose the optimal model (\%Opt.) or chose a misspecified
  model (\%$>d^*$)}
\label{tab:fullSimulation}
\end{table}

\end{document}